\documentclass[11pt]{article}
\usepackage{amsmath,amssymb,amsthm}
\usepackage[margin=1in]{geometry}
\usepackage{booktabs}
\usepackage{hyperref}
\usepackage{iftex}
\ifXeTeX\usepackage{fontspec}\else\usepackage[T1]{fontenc}\usepackage{lmodern}\fi

\newtheorem{theorem}{Theorem}
\newtheorem{lemma}[theorem]{Lemma}
\newtheorem{corollary}[theorem]{Corollary}
\newtheorem{proposition}[theorem]{Proposition}
\theoremstyle{definition}
\newtheorem{definition}[theorem]{Definition}
\theoremstyle{remark}
\newtheorem{remark}[theorem]{Remark}

\newcommand{\Z}{\mathbb{Z}}
\newcommand{\F}{\mathbb{F}}
\newcommand{\GR}{\mathrm{GR}}
\newcommand{\wt}{\mathrm{wt}}
\newcommand{\Tr}{\mathrm{Tr}}
\newcommand{\CAN}{\mathrm{CAN}}
\newcommand{\cT}{\mathcal{T}}
\newcommand{\cP}{\mathcal{P}}
\newcommand{\ol}[1]{\overline{#1}}

\title{Quaternary Preparata codes are quasi-perfect covering codes\\ in the Hamming metric}
\author{Tamal Maharaj\thanks{Ramakrishna Mission Vivekananda Educational and Research Institute, Belur Math, Howrah, India. Corresponding author: \texttt{tamal@gm.rkmvu.ac.in}.}}
\date{}

\begin{document}
\maketitle

\begin{abstract}
The quaternary Preparata codes of Hammons, Kumar, Calderbank, Sloane and Solé are $\Z_4$-linear codes of length $2^m$ with $4^{2^m-m-1}$ codewords. They are usually studied through their binary Gray images, that is, in the Lee metric. We consider them instead as codes over a $4$-letter alphabet with the Hamming metric. We show that for odd $m$ they have covering radius $2$ and minimum distance $4$, so they are quasi-perfect. For even $m$ the same construction has covering radius $3$. The proof is a short computation with Teichmüller representatives in the Galois ring $\GR(4,m)$, and it also gives the coset weight distribution. The covering density of these codes tends to $9/8$. The best previously known infinite families of quaternary codes with covering radius $2$ have covering density about $1.5$. Consequently $K_4(n,2)\le 4^{\,n-m-1}$ for all $n\ge 2^m$ and odd $m$, which improves the best available bounds for a range of lengths above each $2^m$. For $m=3$ the code is the octacode, and it gives $K_4(8,2)\le 256$; the previous bound was $352$, from 2005, and the best lower bound is $251$. We also give a quaternary code showing $K_4(6,2)\le 48$, where the previous bound was $52$, together with consequences for larger alphabets.
\end{abstract}

\section{Introduction}

Let $K_q(n,R)$ be the smallest size of a $q$-ary code of length $n$ with covering radius $R$: every word of $\Z_q^n$ lies within Hamming distance $R$ of some codeword~\cite{chll}. Writing $V_q(n,R)=\sum_{i\le R}\binom ni (q-1)^i$ for the size of a Hamming ball, the \emph{covering density} of a code $C$ with covering radius $R$ is
\[
 \mu(C)=\frac{|C|\,V_q(n,R)}{q^n}\ \ge 1 ,
\]
with equality exactly for perfect codes. For a linear code of codimension $r$ this is $V_q(n,R)/q^r$. The \emph{length function} $\ell_q(r,R)$ is the smallest length of a $q$-ary linear code with codimension $r$ and covering radius $R$. Whenever $n\ge\ell_q(r,R)$ it gives $K_q(n,R)\le q^{n-r}$. A long line of work constructs infinite families with small asymptotic density. For $R=2$ these families are built from saturating sets in projective spaces; see~\cite{dgmp} and the references there. For $q=4$ and $R=2$ the best known infinite families have asymptotic covering density about $1.504$ for even codimension and $1.587$ for odd codimension~\cite[Table~VI]{dgmp}. It is known that $\ell_4(4,2)=9$~\cite[Proposition~5.1]{dgmp}.

$\Z_4$-linear codes became prominent when Hammons, Kumar, Calderbank, Sloane and Solé~\cite{hkcss} showed that the Kerdock, Preparata and Goethals codes are Gray images of linear codes over $\Z_4$. The covering radii of these $\Z_4$-codes have been studied in the Lee and Euclidean metrics~\cite{aoki,hkcss,hk}; the Lee covering radius is the binary covering radius of the Gray image. Aoki, Gaborit, Harada, Ozeki and Solé~\cite{aoki} note that covering radius can be defined for the Hamming, Lee and Euclidean distances, but they deliberately consider only the last two. As far as we can tell, the Hamming covering radius of these codes, viewed as codes over a $4$-letter alphabet, has not been determined.

\paragraph{Main result.} For $m\ge2$ let $\cP_m$ be the $\Z_4$-linear code of length $N=2^m$ defined in Section~\ref{sec:codes}. For odd $m$ it is the quaternary Preparata code of~\cite{hkcss}, and $\cP_3$ is the octacode.

\begin{theorem}\label{thm:intro}
The code $\cP_m$ has minimum Hamming distance $4$. Its Hamming covering radius is $2$ if $m$ is odd and $3$ if $m$ is even. For odd $m$ its covering density is $(9N^2-3N+2)/(8N^2)$, which tends to $9/8$.
\end{theorem}

Theorem~\ref{thm:main} below gives the full coset weight distribution. The consequences are as follows.
\begin{itemize}
\item \textbf{Density.} For odd $m$, $\cP_m$ has length $2q^{(r-2)/2}$ and codimension $r=m+1$, where $q=4$. This is the leading term of the best known $\F_4$-linear families, but without their lower-order terms. The density falls from about $1.504$ to $9/8$ (Section~\ref{sec:density}).
\item \textbf{New upper bounds.} $K_4(n,2)\le 4^{n-m-1}$ for every odd $m\ge3$ and every $n\ge 2^m$. For $2^m\le n<(2^{m+2}-2)/3$ this is smaller than any direct sum of two codes with covering radius $1$, even perfect ones. It also improves the best bounds we could derive from published constructions for $n=32,\dots,36$, $128,\dots,153$ and $512,\dots,591$ (Section~\ref{sec:bounds}).
\item \textbf{The octacode.} $251\le K_4(8,2)\le 256$. The previous upper bound was $352$~\cite{keri-tables}, and the lower bound is from~\cite{gp}. No $\F_4$-linear code can achieve this, because $\ell_4(4,2)=9$ (Section~\ref{sec:octa}). Via a product construction of Kéri and Östergård~\cite{ko} this also improves $K_8(8,2)$, $K_{12}(8,2)$, $K_{16}(8,2)$ and $K_{20}(8,2)$.
\item \textbf{Small cases.} A computer-found code gives $K_4(6,2)\le 48$, improving $52$~\cite{ost99} (Section~\ref{sec:62}).
\end{itemize}

The octacode case was found first, by a computer search for covering codes with a prescribed symmetry group (Section~\ref{sec:methods}). Recognizing the code as $\cP_3$ led to Theorem~\ref{thm:intro}.

\section{Galois rings and Teichmüller representatives}\label{sec:prelim}

We use standard facts about Galois rings; see~\cite{hkcss,wan}. Let $m\ge1$ and let $\GR=\GR(4,m)=\Z_4[x]/(h(x))$, where $h$ is monic of degree $m$ and irreducible modulo $2$. Then $|\GR|=4^m$, and $2\GR$ is the unique maximal ideal. The residue field $\F=\GR/2\GR$ is isomorphic to $\F_{2^m}$; we write $c\mapsto\ol c$ for the reduction map. The \emph{Teichmüller set} $\cT=\{c\in\GR: c^{2^m}=c\}$ consists of $0$ and a cyclic group of order $2^m-1$. Reduction restricts to a bijection $\cT\to\F$; let $\tau:\F\to\cT$ be its inverse. Since $\cT$ is closed under multiplication, $\tau(\alpha\beta)=\tau(\alpha)\tau(\beta)$. Every $c\in\GR$ has a unique \emph{$2$-adic expansion} $c=a+2b$ with $a,b\in\cT$, and here $a=\tau(\ol c)$. Since $2\cdot2\GR=0$, the element $2c$ depends only on $\ol c$: namely $2c=2\tau(\ol c)$. In particular $2a+2b=2\tau(\ol a+\ol b)$ for all $a,b\in\GR$.

\begin{lemma}\label{lem:teich}
For $a,b\in\cT$ with $\alpha=\ol a$ and $\beta=\ol b$,
\[
 a+b=\tau(\alpha+\beta)+2\tau\bigl(\sqrt{\alpha\beta}\bigr).
\]
\end{lemma}

\begin{proof}
For $a,b\in\cT$ write $a+b=\tau(\alpha+\beta)+2d$ with $d\in\cT$, and put $\kappa(a,b)=\ol d$. Squaring and working modulo $4$ gives $(a+b)^2=a^2+b^2+2ab$ and $(\tau(\alpha+\beta)+2d)^2=\tau(\alpha+\beta)^2=\tau(\alpha^2+\beta^2)$. Since $a^2,b^2\in\cT$, we also have $a^2+b^2=\tau(\alpha^2+\beta^2)+2d'$ with $\ol{d'}=\kappa(a^2,b^2)$. Comparing, $2(d'+ab)=0$, so $\kappa(a^2,b^2)=\alpha\beta$. Squaring permutes $\cT$, because it permutes $\F$ and commutes with $\tau$. Applying the identity to the Teichmüller square roots of $a$ and $b$ therefore gives $\kappa(a,b)=\sqrt{\alpha\beta}$.
\end{proof}

\begin{corollary}\label{cor:teich}
For $a,b\in\cT$ with $\alpha=\ol a$ and $\beta=\ol b$,
\begin{align*}
 a-b&=\tau(\alpha+\beta)+2\tau\bigl(\sqrt{\alpha\beta}+\beta\bigr),\\
 -(a+b)&=\tau(\alpha+\beta)+2\tau\bigl(\sqrt{\alpha\beta}+\alpha+\beta\bigr).
\end{align*}
\end{corollary}

\begin{proof}
Use $-b=b+2b$ and $-(a+b)=(a+b)+2(a+b)$, together with Lemma~\ref{lem:teich} and the rule $2c+2c'=2\tau(\ol c+\ol{c'})$.
\end{proof}

We also need the following standard fact about quadratic equations in characteristic~$2$; see~\cite[Theorem~2.25]{ln}. Here $\Tr:\F\to\F_2$ is the absolute trace. For $\xi\in\F^*$ and $c\in\F$, the polynomial $Z^2+\xi Z+c$ has a root in $\F$ if and only if $\Tr(c/\xi^2)=0$. In that case it has two distinct roots, whose sum is $\xi$. We apply this with $c=\eta^2$, where $\Tr(\eta^2/\xi^2)=\Tr(\eta/\xi)$.

\section{The codes \texorpdfstring{$\cP_m$}{Pm} and their covering radius}\label{sec:codes}

\begin{definition}
Let $m\ge2$ and $N=2^m$. Index coordinates by $\cT$, and define the \emph{syndrome map}
\[
 \sigma:\Z_4^{\cT}\to\Z_4\times\GR,\qquad \sigma(x)=\Bigl(\sum_{t\in\cT}x_t,\ \sum_{t\in\cT}x_t\,t\Bigr).
\]
Set $\cP_m=\ker\sigma$. As $\GR\cong\Z_4^m$ as a $\Z_4$-module, $\cP_m$ is the $\Z_4$-linear code with parity-check matrix whose columns are $(1,t)$, $t\in\cT$, with $t$ written in a $\Z_4$-basis of $\GR$.
\end{definition}

\begin{remark}
Let $\mathrm{T}:\GR\to\Z_4$ be the generalized trace. The code $\{(\mathrm{T}(\lambda t)+\varepsilon)_{t\in\cT}:\lambda\in\GR,\ \varepsilon\in\Z_4\}$ is the quaternary Kerdock code of~\cite{hkcss} (for odd $m$). The trace form is non-degenerate, so its dual is exactly $\cP_m$, the quaternary Preparata code of~\cite{hkcss}. For $m=3$ it is the octacode, and an explicit generator matrix is given in Section~\ref{sec:octa}. For even $m$ the Gray image of $\cP_m$ is not a Preparata-type code, but $\cP_m$ is still a well-defined $\Z_4$-linear code.
\end{remark}

For $x\in\Z_4^{\cT}$ let $\wt(x)$ be its Hamming weight. The \emph{weight} of a coset of $\cP_m$ is the smallest weight of its elements, and the covering radius is the largest coset weight. Write $e_t$ for the unit vector at coordinate $t$.

\begin{theorem}\label{thm:main}
Let $m\ge2$ and $N=2^m$.
\begin{enumerate}
\item[(a)] The map $\sigma$ is onto. Hence $|\cP_m|=4^{N-m-1}$ and $\cP_m$ has $4N^2=4^{m+1}$ cosets.
\item[(b)] If $m$ is odd, the numbers of cosets of weight $0,1,2$ are $1,\ 3N,\ 4N^2-3N-1$. So $\cP_m$ has covering radius $2$.
\item[(c)] If $m$ is even, the numbers of cosets of weight $0,1,2,3$ are $1,\ 3N,\ 4N^2-3N-1-\binom N2,\ \binom N2$. So $\cP_m$ has covering radius $3$.
\item[(d)] $\cP_m$ has minimum Hamming distance $4$. If $m$ is odd, its codewords of weight $4$ are exactly the $N(N-1)(N-2)/24$ vectors $2\sum_{t\in S}e_t$, where $S$ runs over the $4$-subsets of $\cT$ with $\sum_{t\in S}\ol t=0$.
\end{enumerate}
In particular, for odd $m$ the code $\cP_m$ is quasi-perfect in the Hamming metric: its packing radius is $1$ and its covering radius is $2$.
\end{theorem}

\begin{proof}
The vectors of weight at most $2$ are $0$; the vectors $ae_t$; and the vectors $ae_t+be_u$ with $t\ne u$, where $a,b\in\{1,2,3\}$. Their syndromes are $(a,at)$ and $(a+b,\,at+bu)$. For each $s\in\Z_4$ we find which syndromes $(s,x)$ are reached, and how often. Throughout, $x=\tau(\xi)+2\tau(\eta)$ is the $2$-adic expansion of $x\in\GR$, and $\alpha=\ol t$, $\beta=\ol u$.

\emph{Case $s=1$.} The vectors are $e_u$, with syndrome $(1,u)$, and $3e_u+2e_t$ with $t\ne u$, with syndrome $(1,-u+2t)$. Since $u=-u+2u$, the reached values of $x$ are exactly $-u+2t$ with $u,t\in\cT$. The weight is $1$ if $t=u$ and $2$ otherwise. Given $x$, the condition $x\equiv-u\equiv u\pmod2$ forces $u=\tau(\ol x)$. Then $x+u\in2\GR$, so $x+u=2t$ for a unique $t\in\cT$. Hence every $(1,x)$ is reached by exactly one vector of weight at most $2$. This gives $N$ cosets of weight $1$ (those with $x\in\cT$) and $N^2-N$ of weight $2$. The map $x\mapsto -x$ preserves weights and negates syndromes, so the case $s=3$ is identical.

\emph{Case $s=0$.} The vectors are $0$; $e_t-e_u$ with $t\ne u$, with syndrome $(0,t-u)$; and $2e_t+2e_u$ with $t\ne u$, with syndrome $(0,2\tau(\alpha+\beta))$. The last kind reach every nonzero element of $2\GR$; for instance take $u=0$. For $x\notin2\GR$ we have $\xi\ne0$. By Corollary~\ref{cor:teich}, $t-u=x$ holds if and only if $\alpha+\beta=\xi$ and $\sqrt{\alpha\beta}+\beta=\eta$. Squaring the second equation and substituting $\alpha=\xi+\beta$ gives $\xi\beta=\eta^2$. Conversely, $\beta=\eta^2/\xi$ and $\alpha=\xi+\beta$ satisfy both equations. So $(t,u)\mapsto t-u$ is a bijection from ordered pairs of distinct elements of $\cT$ onto $\GR\setminus2\GR$. Every $(0,x)$ is reached. There is one coset of weight $0$, the other $N^2-1$ cosets have weight $2$, and no vector of weight $1$ has $s=0$.

\emph{Case $s=2$.} The vectors are $2e_t$, with syndrome $(2,2t)$, reaching all of $2\GR$ exactly once; $e_t+e_u$, with syndrome $(2,t+u)$; and $3e_t+3e_u$, with syndrome $(2,-(t+u))$. Here $t\ne u$, so these last two lie outside $2\GR$. Let $x\notin2\GR$. By Lemma~\ref{lem:teich}, $x=t+u$ exactly when $\alpha,\beta$ are the roots of $Z^2+\xi Z+\eta^2$. By the fact at the end of Section~\ref{sec:prelim}, such $t,u$ exist if and only if $\Tr(\eta/\xi)=0$, and then the pair $\{t,u\}$ is unique. By Corollary~\ref{cor:teich}, $x=-(t+u)$ exactly when $\alpha,\beta$ are the roots of $Z^2+\xi Z+(\eta+\xi)^2$. Such $t,u$ exist if and only if $\Tr(\eta/\xi)+\Tr(1)=0$. Now $\Tr(1)\equiv m\pmod 2$.

If $m$ is odd, exactly one of the two conditions holds. So every $(2,x)$ is reached by exactly one vector of weight at most $2$: there are $N$ cosets of weight $1$ and $N^2-N$ of weight $2$.

If $m$ is even, the two conditions coincide. The $(N-1)N/2=\binom N2$ syndromes $(2,x)$ with $\xi\ne0$ and $\Tr(\eta/\xi)=1$ are then not reached by any vector of weight at most $2$. For such $x$, the case $s=0$ gives distinct $t,u\in\cT$ with $t-u=x$. Then $2e_0+e_t+3e_u$ has syndrome $(2,x)$. It has weight $3$, because if $0\in\{t,u\}$ it would have weight at most $2$, which we have excluded.

Collecting the cases proves (b) and (c). Every syndrome is reached, so (a) holds.

For (d), let $c\ne0$ be a codeword with $\wt(c)\le4$. Split its support into two parts of size at most $2$. This writes $c=v-w$ with $\wt(v),\wt(w)\le2$, $v\ne w$ and $\sigma(v)=\sigma(w)$. By the analysis above, such coincidences occur only in two situations. The first is two vectors $2e_t+2e_u$ and $2e_{t'}+2e_{u'}$ with $\ol t+\ol u=\ol{t'}+\ol{u'}$; the two pairs are then disjoint. The second, only for even $m$, is $e_t+e_u$ and $3e_{t'}+3e_{u'}$ with the same syndrome. These pairs are also disjoint: if $t=t'$, then $2t+u+u'=0$, which forces $\ol u=\ol{u'}$ and then $t=u$. In every case $\wt(c)=4$, so the minimum distance is $4$. For odd $m$ only the first situation occurs. It gives exactly the vectors $2\sum_{t\in S}e_t$ with $|S|=4$ and $\sum_S\ol t=0$. The number of $4$-subsets of $\F$ with zero sum is $N(N-1)(N-2)/24$.
\end{proof}

\begin{remark}
For odd $m$ the proof shows that every coset of $\cP_m$ has a unique leader of weight at most $2$, except the $N-1$ cosets with syndrome $(0,x)$, $0\ne x\in2\GR$. Each of those contains $N/2$ leaders of the form $2e_t+2e_u$. This accounts for the whole excess $V_4(N,2)-4N^2=(N-1)(N-2)/2$. The minimum Hamming weight $4$ is consistent with~\cite{iy}, which treats generalized Hamming-type codes over Galois rings. The parity of $m$ enters only through $\Tr(1)$. A parity condition also appears in a packing problem. Honold and Landjev~\cite[Remark~3.7]{hl} showed that the Teichmüller set of a Galois ring of characteristic $4$ and degree $f$ is a cap in the corresponding projective Hjelmslev geometry if and only if $f$ is odd.
\end{remark}

\section{Covering density}\label{sec:density}

For odd $m$, Theorem~\ref{thm:main} gives
\[
 \mu(\cP_m)=\frac{V_4(N,2)}{4^{m+1}}=\frac{1+3N+\tfrac92N(N-1)}{4N^2}=\frac{9N^2-3N+2}{8N^2}\ \longrightarrow\ \frac98 .
\]
The code has quaternary codimension $r=m+1$ (even) and length $2^{r-1}=2q^{(r-2)/2}$ with $q=4$. Table~\ref{tab:density} compares it with the best known quaternary codes of the same codimension and covering radius $2$, which are $\F_4$-linear~\cite{dgmp}. The asymptotic density of the best known families is $1.504$ for even and $1.587$ for odd codimension~\cite[Table~VI]{dgmp}. We found no later improvement for $q=4$.

\begin{table}[ht]
\centering
\begin{tabular}{cccccc}
\toprule
$r$ & sphere bound & best known $n$~\cite{dgmp} & density & $\cP_{r-1}$: $n$ & density\\
\midrule
4 & 8 & 9 (exact) & 1.375 & 8 & 1.082\\
6 & 31 & 37 & 1.491 & 32 & 1.114\\
8 & 121 & 154 & 1.625 & 128 & 1.122\\
10 & 483 & 592 & 1.503 & 512 & 1.124\\
\bottomrule
\end{tabular}
\caption{Quaternary codes with covering radius $2$ and $4^r$ cosets (codimension $r$). The sphere bound is the least $n$ with $V_4(n,2)\ge4^r$.}\label{tab:density}
\end{table}

\begin{remark}
Binary codes with covering radius $2$ and density tending to $1$ are known~\cite{struik}, and the binary family of Etzion and Greenberg~\cite{eg} has exactly the densities of Table~\ref{tab:density}, since $V_2(3n-1,2)=V_4(n,2)$. Mapping $\cP_m$ to binary codes of length $3\cdot2^m$ by the standard construction~\cite[Lemma~4.74]{struik} gives nothing better than these. Our contribution is to the quaternary case. There the construction appears to be special to the prime $2$. Even if the analogous codes over $\GR(p^2,m)$ had covering radius $2$, their density would be about $(p^2-1)^2/(2p^2)$. This is $32/9$ for $p=3$, well above the densities $1.707$ and $1.782$ of the known families for $q=9$~\cite[Table~VI]{dgmp}.
\end{remark}

\section{New upper bounds on \texorpdfstring{$K_4(n,2)$}{K4(n,2)}}\label{sec:bounds}

\begin{corollary}\label{cor:bounds}
For every odd $m\ge3$ and every $n\ge 2^m$,
\[
 K_4(n,2)\le 4^{\,n-m-1}.
\]
\end{corollary}

\begin{proof}
Take the direct sum of $\cP_m$ with the full space $\Z_4^{\,n-2^m}$.
\end{proof}

To our knowledge no table of bounds on $K_4(n,2)$ goes beyond $n=11$. For $n\ge 12$ one compares with what can be derived from published constructions. Linear codes of codimension $m+1$ with covering radius $2$ are known only for $n\ge 37,154,592$ when $m=5,7,9$ (Table~\ref{tab:density}). Below these lengths the best linear bound is $4^{n-m}$, four times larger. Direct sums of two codes with covering radius $1$ cannot compete near $n=2^m$ either. If $n=n_1+n_2$, the sphere bound gives
\[
 K_4(n_1,1)K_4(n_2,1)\ \ge\ \frac{4^n}{(1+3n_1)(1+3n_2)}\ \ge\ \frac{4^n}{(1+3n/2)^2},
\]
and this exceeds $4^{n-m-1}$ whenever $n<(2^{m+2}-2)/3$. For $n=32,\dots,36$ the best bound we could derive otherwise is $2\cdot4^{n-6}$. It comes from the direct sum of the perfect $[21,18]_4$ Hamming code with Kéri's code showing $K_4(11,1)\le131072$~\cite{keri-tables}, followed by appending coordinates. Corollary~\ref{cor:bounds} halves this.

\section{The octacode and \texorpdfstring{$K_4(8,2)$}{K4(8,2)}}\label{sec:octa}

For $m=3$, $\cP_3$ is the octacode. One generator matrix is $[I_4\mid A]$ with
\[
 A=\begin{pmatrix}1&3&1&2\\1&2&3&1\\1&1&2&3\\2&3&3&3\end{pmatrix}.
\]
This code is self-dual with Hamming weight enumerator $1+14z^4+112z^5+112z^7+17z^8$. By Theorem~\ref{thm:main} its cosets have weights $0,1,2$ with multiplicities $1,24,231$.

\begin{corollary}
$251\le K_4(8,2)\le256$.
\end{corollary}

In Kéri's tables the previous bound was $240\le K_4(8,2)\le352$~\cite{keri-tables}. The upper bound $352$ is an unpublished computer result from 2005, due to Bertolo, Di Pasquale and Santisi. The lower bound was raised to $251$ by Gijswijt and Polak~\cite{gp}, who still list $352$ as the best known upper bound. No $\F_4$-linear code can give Corollary~\ref{cor:bounds} for $m=3$: a code of length $8$ and codimension $4$ with covering radius $2$ would correspond to an $8$-point $1$-saturating set in $\mathrm{PG}(3,4)$, and $\ell_4(4,2)=9$~\cite[Proposition~5.1]{dgmp}. Davydov et al.\ also remark that no codes of codimension $4$ with covering radius $2$ and length $n<2q+1$ seem to be known~\cite{dgmp}. The octacode has length $2q=8$.

\paragraph{Larger alphabets.} Kéri and Östergård~\cite{ko} proved $K_{tq}(n,R)\le\sigma_t(n,n-R)\,K_q(n,R)$. Here $\sigma_t(n,k)=\CAN(k,n,t)$ is the size of the smallest $t$-ary covering array of strength $k$ with $n$ columns. The previous entry $K_8(8,2)\le29920=85\cdot352$ came from this rule with $\CAN(6,8,2)=85$~\cite{keri-tables}. With $256$ in place of $352$ we obtain the following.

\begin{corollary}
$K_8(8,2)\le21760$, $K_{12}(8,2)\le290304$, $K_{16}(8,2)\le1441792$ and $K_{20}(8,2)\le6400000$.
\end{corollary}

These use $\CAN(6,8,t)\le85,\,1134,\,5632,\,25000$ for $t=2,3,4,5$~\cite{colbourn}. The previous upper bounds were $29920$, $323040$, $1954905$ and $7929856$~\cite{keri-tables}.

\section{A new bound for \texorpdfstring{$K_4(6,2)$}{K4(6,2)}}\label{sec:62}

\begin{proposition}
$33\le K_4(6,2)\le48$.
\end{proposition}

The lower bound is from~\cite{gp}; the previous upper bound was $52$~\cite{ost99}. The $48$ words of a code attaining the upper bound are listed in Table~\ref{tab:62}. The code is invariant under the negation $x\mapsto -x\pmod4$ and under the XOR translations by $\{000000,\,133033,\,222022,\,311011\}$, where each symbol is read as an element of $\F_2^2$. It is the union of $7$ orbits of the group these generate. It is not a union of cosets of any nontrivial $\Z_4$-linear code: its stabiliser under $\Z_4$-translations has order $2$. We have no conceptual proof. The covering radius was checked by two independent programs (Section~\ref{sec:methods}).

\begin{table}[ht]
\centering\small
\begin{tabular}{llllllll}
001220&002220&003220&010112&010303&011031&020133&020311\\
022002&030101&030332&033013&100020&103132&103301&111031\\
113100&113322&122002&123121&123330&130213&131213&132213\\
200020&202111&202333&211031&212123&212310&220202&221202\\
223202&232130&232321&233013&300020&301103&301312&310231\\
312231&313231&321110&321323&322002&331122&331300&333013
\end{tabular}
\caption{A quaternary code of length $6$, size $48$ and covering radius $2$.}\label{tab:62}
\end{table}

\section{Computations}\label{sec:methods}

\paragraph{Search.} The octacode and the code of Table~\ref{tab:62} were found by a tabu search for covering codes that are unions of orbits of a prescribed group $G$ of Hamming isometries. The search is in the style of Östergård, run over all open cells of Kéri's tables with $q\le4$. A simple learned policy chose $G$ from a menu of generator families: translations, coordinate permutations and symbol maps. The same search with $G$ trivial did not beat $352$ for $K_4(8,2)$ in $24$ hours on $40$ cores. After the octacode was recognized, a second search looked for codes that are $\Z_4$-linear or additive, or unions of their cosets. It covered all $29$ open quaternary cells with $n\le11$, and it found no further improvement.

\paragraph{Verification.} Every code was checked by a program that shares no code with the search. It computes the distance from every word to the code by multi-source breadth-first search in the Hamming graph. For the codes $\cP_m$ we also computed the coset weight distribution directly for $3\le m\le9$, by breadth-first search in the Cayley graph of $\Z_4\times\GR(4,m)$ generated by the syndromes of weight-one vectors. The results agree with Theorem~\ref{thm:main}: for $m=3,5,7,9$ the weights are $(1,24,231)$, $(1,96,3999)$, $(1,384,65151)$ and $(1,1536,1047039)$, and for $m=4,6,8$ the weights include $120$, $2016$ and $32640$ cosets of weight $3$. Part~(d) was checked by enumerating all vectors of weight at most $4$ for $m=3,4,5$. For odd $m$ this found exactly the predicted $14$ and $1240$ codewords of weight $4$. Lemma~\ref{lem:teich} was also checked numerically for $3\le m\le7$. All codes and scripts are available from the author.

\section{Concluding remarks}

\begin{enumerate}
\item For even $m$ the code $\cP_m$ fails to have covering radius $2$ only on the $\binom N2$ syndromes $(2,x)$ with $\Tr(\eta/\xi)=1$. A modification that repairs these at small cost in length would give quaternary codes of odd codimension with density close to $9/8$. Unstructured search does not seem to help, because the codes $\cP_m$ are very special. Hill-climbing over $\Z_4$-linear codes with $4^6$ cosets and length $32$ never came close to $\cP_5$: the best code it found left $516$ of the $4096$ syndromes uncovered.
\item Other classical $\Z_4$-codes do not behave the same way. For the quaternary Goethals codes, with parity-check columns $(1,t,2t^3)$, we computed Hamming covering radius $5$ for $m=3,5$.
\item It would be interesting to know whether $9/8$ can be beaten for $q=4$ and $R=2$, for instance by blockwise direct-sum constructions as in the binary case~\cite{struik}.
\end{enumerate}

\section*{Acknowledgements}
The computer searches, the identification of the octacode as the first member of this family, the proof of Theorem~\ref{thm:main} and a first draft of this paper were produced with the assistance of an AI system (Claude, Anthropic). All computational claims were verified by independent programs. The author takes full responsibility for the content of this paper.

\end{document}